\documentclass[12pt]{article}

\usepackage{amsmath} \usepackage{amsfonts}
\usepackage{amssymb} \usepackage{latexsym} \usepackage{enumerate}
\usepackage{mathtools}
\mathtoolsset{showonlyrefs}

 \usepackage{mathrsfs}
\usepackage[numbers]{natbib}

\input{style.tex}

\renewcommand{\hat}{\widehat}

\newtheorem{thm}{Theorem}[section]

\newtheorem{lem}[thm]{Lemma}
\newtheorem{rem}[thm]{Remark}
\newtheorem{cor}[thm]{Corollary}

\begin{document}

\title{A Note on Utility Indifference Pricing with Delayed Information}

\author{ Peter Bank\footnote{Technische Universit{\"a}t Berlin,
    Institut f{\"u}r Mathematik, Stra{\ss}e des 17. Juni 136, 10623
    Berlin, Germany, email \texttt{bank@math.tu-berlin.de}. P.Bank is supported in part by the GIF Grant 1489-304.6/2019}
  \hspace{2ex} Yan Dolinsky\footnote{ Department of Statistics, Hebrew
    University.
    email \texttt{yan.dolinsky@mail.huji.ac.il}. Y. Dolinsky is supported in part by the GIF Grant 1489-304.6/2019 and the ISF grant 160/17 }}

\date{\today}

\maketitle
\begin{abstract}
   We consider the Bachelier model with information delay where investment decisions can be based only on observations from $H>0$ time units before. Utility indifference prices are studied for vanilla options and we compute their non-trivial scaling limit for vanishing delay when risk aversion is scaled liked $A/H$ for some constant $A$. Using techniques from~\cite{Fritelli:00}, we develop discrete-time duality for this setting and show how the relaxed form of martingale property introduced by~\cite{KS:06} results in the scaling limit taking the form of a volatility control problem with quadratic penalty.
\end{abstract}

\begin{description}
\item[Mathematical Subject Classification (2010):] 91G10, 91G20

\item[Keywords:] Utility Indifference Pricing, Hedging with Delay, Asymptotic Analysis
\end{description}





\maketitle

\markboth{P. Bank and Y. Dolinsky}{Utility Indifference Pricing with Delayed Information}
\renewcommand{\theequation}{\arabic{section}.\arabic{equation}}
\pagenumbering{arabic}
\section{Introduction and the Main Result}

Taking into account frictions is an important challenge in financial modelling.
In this paper, we focus on the friction that investment decisions may be based only on delayed information,
the actual present market prices being unknown at the time of decision making. In line with other, more widely
investigated frictions like transaction costs, information delay or, more generally, operating under partial information makes the computation
of prices and hedges rather difficult. It is thus of interest to resort to asymptotic methods that study suitably
re-scaled models for vanishing friction.

For the case of transaction costs,
the seminal work of Barles and Soner \cite{BS:98} determines the scaling limit
of utility indifference prices of vanilla options for small proportional transaction costs and high risk aversion.
The present note provides an analogous analysis for the case of information delay,
albeit using convex duality and martingale techniques rather than taking a
PDE approach as pursued in~\cite{BS:98}.

Let us briefly review some of the relevant literature on models with partial information. In \cite{S:94} the author
solved the mean variance hedging problem with partial information for European contingent claims when the asset price process is a martingale. These results were extended beyond the martingale case, for instance, in \cite{F:00,MTT:08} and for defaultable claims in \cite{KX:07}. In a recent paper
\cite{MZS:18} the authors studied the mean variance hedging problem in a setup where the investor has deterministic information
and obtained explicit
solutions in the Bachelier and Black–Scholes models.
The papers \cite{KS:06, CKT:17} develop an arbitrage theory
for hedging with partial information and the papers
\cite{IM:17,DZ:20} treat super-replication with delayed information.
For more details about hedging with partial information see also \cite{Z:17}.
Let us also mention the work
\cite{NO:09} which applied Malliavin calculus
to the utility maximization problem with delayed information. A more recent paper is
\cite{SZ:19} where the authors introduce and study a general framework of optimal stochastic control with delayed information.

To formulate our main result in this note, let
$(\Omega, \mathcal{F}, \mathbb P)$ be  a complete probability space carrying a
one--dimensional Brownian motion
$(W_t)_{t \in [0,T]}$ with natural augmented filtration $(\mathcal{F}^W_t)_{t \in [0,T]}$ and time horizon $T \in (0,\infty)$.
We consider a simple financial market
with a riskless savings account bearing zero interest (for simplicity) and with a risky
asset $S$ with Bachelier price dynamics
\begin{equation*}
S_t=s_0+\sigma W_t+\mu t, \quad t \in [0,T],
\end{equation*}
where $s_0 \in \RR$ is the initial asset price, $\sigma>0$ is the constant volatility and $\mu\in\mathbb{R}$ is the constant drift.

Contrary to the usual setting considered in Mathematical Finance, we focus on an investor who is informed about the risky asset's price changes with a delay $H>0$. As a result, the investor's decisions at time $t$ can
depend only on the information observed up to time $t-H$ as described by the filtration
$\mathcal G^H_t:=\mathcal F^W_{(t-H)^+}$, $t\in [0,T]$.
In line with this delay, the set $\cA^H$ of admissible simple trading strategies consists of
piecewise constant processes of the form
\begin{equation}\label{form}
\gamma_t=\sum_{i=1}^{N} \gamma_{t_i} 1_{(t_{i-1},t_{i}]}(t), \quad  t\in [0,T],
\end{equation}
where $N\in\mathbb N$ is deterministic,
$0=t_0< t_1<\dots< t_N=T$ form a partition of $[0,T]$ and where each
 $\gamma_{t_i}$,  $i=1,\dots,N$, is a $\mathcal G^H_{t_{i-1}}$-measurable random variable.
For any $\gamma\in\mathcal A^H$ of the form~\eqref{form},
the corresponding P\&L at the maturity date $T$ is given as usual:
\begin{equation*}
V^{\gamma}_T:=\sum_{i=1}^{N} \gamma_{t_i}(S_{t_{i}}-S_{t_{i-1}}).
\end{equation*}
The investor will use this investment opportunity to hedge against the payoff of a vanilla European option with payoff $F=f(S_T)$ where
$f:\mathbb R\rightarrow [0,\infty)$ is continuous with strictly sub-quadratic growth in the sense that for some $p\in (0,2)$ there is $C \in [0,\infty)$ with
\begin{equation}\label{2.0+}
0 \leq f(x)\leq C(1+|x|^p),  \quad x\in\mathbb R.
\end{equation}
It is immediate from~\eqref{2.0+} that
\begin{align*}
\mathbb E_{\mathbb P}[\exp\left({\alpha f(S_T)}\right)]<\infty, \quad \alpha\in\mathbb R.
\end{align*}
The investor will assess the quality of a hedge by the resulting expected utility. Assuming exponential utility with constant absolute risk aversion $\lambda>0$, the utility indifference price of
the claim $f(S_T)$ (see, e.g., \cite{R:08} for details on indifference prices) does not depend on the investor's initial wealth
and takes the well-known form
\begin{equation}\label{2.2}
\pi(H,\lambda,f(S_T))=
\frac{1}{\lambda}\log\left(\frac{\inf_{\gamma\in\mathcal A^H}\mathbb E_{\mathbb P}\left[\exp\left(-\lambda\left(V^{\gamma}_T-f(S_T)\right)\right)\right]}
{\inf_{\gamma\in\mathcal A^H}\mathbb E_{\mathbb P}\left[\exp\left(-\lambda V^{\gamma}_T\right)\right]}\right).
\end{equation}
\begin{rem}\label{rem:denom}
 The infimum in the denominator of~\eqref{2.2} is well-known to be attained by the constant strategy $\gamma \equiv \mu/(\lambda \sigma^2)$ which of course is admissible with any delay in information and which, for any risk aversion $\lambda>0$, yields the same value $\exp(-(\mu^2T/(2\sigma^2))$. As a consequence, the relevant quantity to be understood is the numerator in~\eqref{2.2}.
\end{rem}
It is readily seen that for $H \downarrow 0$, the above indifference price converges to the usual indifference price with no delay in information. A more interesting limit emerges, however, if we re-scale the investor's risk-aversion in the form $A/H$ as described by our main result which will be proved in Section~\ref{sec:3}:
\begin{thm}\label{thm.1}
For vanishing delay $H \downarrow 0$ and re-scaled high risk-aversion $A/H$ with $A>0$ fixed, the utility indifference price of $f(S_T)$ with continuous $f \geq 0$ exhibiting sub-quadratic growth~\eqref{2.0+}  has the scaling limit
\begin{equation}\label{2.3}
 \lim_{H\downarrow 0} \pi(H,A/H,f(S_T))=
 \sup_{\nu\in\Gamma}\mathbb E_{\mathbb P}\left[ f\left(s_0+\int_{0}^T \nu_t dW_t\right)- \frac{1}{2A\sigma^2}\int_{0}^T\left(\nu_t-\sigma\right)^2dt \right],
\end{equation}
where $\Gamma$ is the class of all bounded  $\mathcal(\cF^W_t)$-predictable processes $\nu$.
\end{thm}
\begin{rem}
The linear re-scaling of risk aversion $A/H$ is in line with the scaling-limit result for super-replication with delay \cite{DZ:20} where it was shown that the delay has to be re-scaled linearly in the time step.
A more formal reason for the scaling
$\frac{A}{H}$ can be found via duality with entropy minimization and the bound given by Lemma~\ref{lem:entropybound}.
\end{rem}
Next, as a corollary of Theorem \ref{thm.1} we obtain the following super-replication result which is the Bachelier version of Theorem 2.1 in \cite{DZ:20}.
\begin{cor}\label{cor.referee}
For any delay $H>0$, the super--replication price of $f(S_T)$
is given by
$\hat f(s_0)$
where $\hat f:\RR \to [0,\infty]$ denotes
the concave envelope of $f$. Moreover,
if $\hat f(s_0)<\infty$ then the buy-and-hold strategy $\gamma\equiv \partial_{+}\hat{f}(s_0)$ is an optimal super-hedge.
\end{cor}
\begin{proof}
First, if $\hat f(s_0)<\infty$ then from concavity $\hat f(s)<\infty$ everywhere and we obtain
$$\hat f(s_0)+\partial_{+}\hat{f}(s_0)(S_T-s_0)\geq\hat f(S_T)\geq f(S_T) \ \ \mbox{a.s.}$$
In other words the strategy with initial capital $\hat f(s)$ and
$\gamma\equiv \partial_{+}\hat{F}(s)$
is a super-hedge. To complete the proof it thus remains to establish that
the super--replication price for delay $H>0$, which we denote by $\Pi(H)$, is no less than $\hat f(s_0)$. Clearly, $\Pi(H)$ is non-increasing in $H$ and any super-replication price is larger than the corresponding utility indifference price. In particular, we have $\Pi(H)\geq \Pi(h)\geq \pi(h,A/h,F)$ for any $h<H$ and $A>0$.
Theorem \ref{thm.1} allows us to let $h \downarrow 0$ and then $A \uparrow \infty$ in this inequality to conclude
\begin{align*}
 \Pi(H) &\geq \lim_{A \uparrow \infty}  \sup_{\nu\in\Gamma}\mathbb E_{\mathbb P}\left[ f\left(s_0+\int_{0}^T \nu_t dW_t\right)- \frac{1}{2A\sigma^2}\int_{0}^T\left(\nu_t-\sigma\right)^2dt \right]\\
 &=\sup_{\nu\in\Gamma}\mathbb E_{\mathbb P}\left[ f\left(s_0+\int_{0}^T \nu_t dW_t\right)\right]=\hat f(s_0),
\end{align*}
where the last equality is well-known in stochastic control; see, e.g., \cite{N:18}, Lemma~3.3 for a proof.
\end{proof}

\begin{rem}
Our otherwise standing continuity and growth assumptions on $f$
can be relaxed in Corollary~\ref{cor.referee} to just lower-semicontinuity. Indeed, the proposed buy-and-hold strategy is then still a super-hedge. To obtain its optimality, we consider an increasing sequence of continuous and bounded payoff profiles $f_n:\mathbb R\rightarrow [0,\infty)$ converging to $f$ from below. Their super-replication prices are obviously not larger than that of $f$ and, by the already established version of Corollary~\ref{cor.referee}, they are given by their concave envelopes
$\hat{f}_n(s_0)\uparrow\hat{f}(s_0)$.
\end{rem}

The proof of Theorem~\ref{thm.1} rests on the duality theory of exponential utility which we will develop for discrete-time information delayed models in Section~\ref{sec:duality}. As in~\cite{KS:06}, it identifies dual probability measures via a relaxation of the martingale property. A lower bound for the scaling limit is then obtained in Section~\ref{sec:lowerbound} by constructing suitable relaxed martingale measures which asymptotically generate a prescribed volatility profile $\nu$. Of course, the entropy of these measures relative to $\P$ explodes (unless $\nu = \sigma$), but we show in~\eqref{eq:333} that the re-scaled risk aversion compensates this adequately in our scaling limit. An upper bound is established in Section~\ref{sec:upperbound} by a suitable lower bound for the entropy of such a relaxed measure which is argued similar to It\^{o}'s isometry, albeit only aggregating 2nd moments from sufficiently distant increments in a carefully chosen way.

Let us finish this section with the following two remarks on how one could compute the right-hand side of~\eqref{2.3}.
\begin{rem}\begin{enumerate}
\item \emph{Connection with optimal transport.} If the distribution of the random variable
$s+\int_{0}^T \nu_u dW_u$
is fixed then from the It\^{o} Isometry we get that the optimal $\nu$
in the right hand side of~\eqref{2.3} is the one which maximizes (under the constraint on the distribution
of $s+\int_{0}^T \nu_u dW_u$)
the expectation $\mathbb E_{\mathbb P}\left[\int_{0}^T \nu_u du\right]$. From Theorem~1.10
in \cite{BBHK:20} we get that the optimal $\nu$ integrates to a random variable of the form
$\int_{0}^T \nu_u dW_u=\zeta(W_T)$
for some appropriate $\zeta:\mathbb R\rightarrow\mathbb R$.
Thus, letting $\mu\sim W_T$ denote, the normal distribution with mean zero and variance $T$, It\^{o}'s representation theorem isometry show that the right hand side of~\eqref{2.3} can be represented as
$$\sup_{\zeta:\mathbb R\rightarrow\mathbb R \text{ with } \int_{\mathbb R} \zeta(z)\mu(dz)=0 } \int_{\mathbb R} \left(f(s+\zeta(z))-\frac{1}{2A\sigma^2}[\zeta(z)-\sigma z]^2\right)\mu(dz).$$
This is a deterministic optimization problem which can be analyzed with
the classical techniques of calculus of variations.

\item \emph{Connection with nonlinear PDE.} The right-hand side of~\eqref{2.3} can also be viewed as the Bachelier version of the stochastic control problem
which is given by the right-hand side of formula~(3) in
\cite{L:18}. That paper studies
hedging of a European option in the Black--Scholes model with linear price impact and describes the option's replication price  by a stochastic control problem similar to the right hand side of~\eqref{2.3}.
Moreover, the author studies in detail the PDE which corresponds to such a control problem. We have to leave open how to give a financial interpretation of the similarity between the mathematical form of our scaling limit representation~\eqref{2.3} and the
replication price formula of \cite{L:18}.
\end{enumerate}
\end{rem}

\section{Proof of Theorem \ref{thm.1}}\label{sec:3}


The proof of Theorem~\ref{thm.1} will be done by estimating the scaling $\liminf$ from below in Section~\ref{sec:lowerbound} and the scaling $\limsup$ from above in~Section~\ref{sec:upperbound}. Both sections rely on duality results for discrete-time models with information delay that we collect in the preparatory Section~\ref{sec:duality} that we turn to next.

\subsection{Duality in models with information delay}\label{sec:duality}

The proof of Theorem~\ref{thm.1} relies on the duality theory for exponential utility maximization with information delay. We will combine results from \cite{KS:06} and \cite{Fritelli:00} to this end.

Taking an abstract view, let us consider a price process $S$ adapted to a filtration $(\cF_t)$ and an investor who will make investment decision using a filtration $(\cG_t)$ which may be strictly contained in $(\cF_t)$ at any time. The payoff will also be given by a general random variable $F$.

\begin{lem}\label{lem:easydualitybound}
 Let $\gamma$ be a $(\cG_t)_{t \in \TT}$-predictable strategy that changes value only at times in $\TT=\{0=t_0<\dots<t_N\leq T\}$ and suppose that $\E_{\P}[\exp(\lambda |F|)]<\infty$. Then
 \begin{align}\label{eq:53}
     \frac{1}{\lambda} \log \E_{\P}\left[\exp\left(-\lambda(V^\gamma_T-F)\right)\right] \geq \E_{\QQ}\left[F-\frac{1}{\lambda} \log \frac{d\QQ}{d\P}\right]
 \end{align}
 for any probability $\QQ\sim\P$ with finite entropy $\E_{\QQ}[\log \frac{d\QQ}{d\PP}]<\infty$ relative to $\P$ such that for all $u \in \TT$  we have $S_u \in L^1(\QQ)$ and
 \begin{align}\label{eq:51}
     \E_{\QQ}\left[S_t-S_s\middle| \cG_s\right]=0 \text{ for any } s,t \in \TT \text{ with } s \leq t.
 \end{align}
\end{lem}
\begin{rem}
 \cite{KS:06} show that existence of a probability measure $\QQ\sim\P$ with the relaxed martingale property~\eqref{eq:51} is equivalent to the absence of arbitrage when investors are confined to $(\cG_t)$-predictable strategies acting only at times from $\TT$. The relaxation accounts for the information differential between the filtration $(\cF_t)$ to which $S$ is adapted and the filtration $(\cG_t)$ that the investor's trades are based upon.
\end{rem}
As usual, the proof rests on the classical Legendre-Fenchel duality inequality
\begin{align}\label{eq:2}
    xy \leq e^x + y(\log y-1), \quad x \in \RR, y>0, \quad \text{ with `$=$' iff } y=e^x.
\end{align}
For instance, using this inequality with $x=\lambda F$, $y=d\QQ/d\P$ yields
\begin{align*}
    \lambda F \frac{d\QQ}{d\P} \leq e^{\lambda F} + \frac{d\QQ}{d\P}\log \frac{d\QQ}{d\P} \in L^1(\P)
\end{align*}
so that the right-hand side in~\eqref{eq:53} is always well-defined under our assumptions on $F$ and $\QQ$.
\begin{proof}

Without loss of generality we can assume that the expectation on the left-hand side of~\eqref{eq:53} is finite. Note that then by~\eqref{eq:2} and the Cauchy-Schwarz inequality
\begin{align*}
    \E_\QQ\left[\frac{\lambda}{2} (V^\gamma_T)^-\right]
    &=\E_\P\left[\frac{\lambda}{2} (V^\gamma_T)^-\frac{d\QQ}{d\P}\right]
    \leq \E_\P\left[e^{\frac{\lambda}{2} (V^\gamma_T)^-}+\frac{d\QQ}{d\P}\left(\log\frac{d\QQ}{d\P}-1\right)\right]\\    &\leq 1+\E_\P\left[e^{-\frac{\lambda}{2} (V^\gamma_T-F)}e^{-\frac{\lambda}{2}F}\right]+\E_{\QQ}\left[\log\frac{d\QQ}{d\P}-1\right]\\
    &\leq \left(\E_\P\left[e^{-\lambda (V^\gamma_T-F)}\right]\right)^{1/2}\left(\E_{\P}\left[e^{-{\lambda}F}\right]\right)^{1/2}+\E_{\QQ}\left[\log\frac{d\QQ}{d\P}\right]
\end{align*}
is finite under our assumptions on $F$ and $\QQ$.

We will next use a variant of a classical martingale transform argument (see, e.g., \cite{FollmerSchied:16}, Theorem~5.14) to conclude from this that actually $V^\gamma_T$ is $\QQ$-integrable with $\E_\QQ[V^\gamma_T]=0$. Indeed, using (a priori) generalized conditional expectations (as defined in~\cite{S}), we can now deduce by Jensen's inequality that $\hat{V}_{T}:=\E_\QQ\left[V^\gamma_T\middle|\cG_T\right]$ exists with $\hat{V}_{T}^- \in L^1(\QQ)$. Hence, $\hat{V}_{t_{N-1}}:=\E_\QQ\left[\hat{V}_{t_N}\middle|\cG_{t_{N-1}}\right]$ exists which, due to~\eqref{eq:51}, can be seen to be equal to $\E_\QQ\left[V^\gamma_{t_{N-1}}\middle|\cG_{t_{N-1}}\right]$; by Jensen's inequality we have again $\hat{V}_{t_{N-1}}^-\in L^1(\QQ)$. We can now proceed inductively to conclude that $\hat{V}_0:=\E_\QQ\left[\hat{V}_{t_1}\middle|\cG_0\right]$ exists and coincides with
$\E_{\QQ}[V^\gamma_0]=0$. Reversing the course of our induction argument, it follows that $\hat{V}_{t_1}$ is $\QQ$-integrable with vanishing expectation, a property that inductively follows finally also for $\hat{V}_{T}$ and, thus, for $V^\gamma_T$ itself.

To conclude our claim~\eqref{eq:53} we can now use~\eqref{eq:2} again to find for any $y>0$ the estimate
\begin{align*}
    \E_{\P}\left[\exp\left(-\lambda(V^\gamma_T-F)\right)\right]
    &=\E_{\P}\left[\left(e^{-\lambda V^\gamma_T}+\lambda V^\gamma_T y e^{-\lambda F} \frac{d\QQ}{d\P}\right)e^{\lambda F}\right] \\
    &\geq \E_{\P}\left[\left(-y e^{-\lambda F}\frac{d\QQ}{d\P}\left(\log \left(ye^{-\lambda F} \frac{d\QQ}{d\P}\right)-1\right)\right)e^{\lambda F}\right] \\
    &= -y\left(\log y -1\right)+\E_{\QQ}\left[\lambda F-\log \frac{d\QQ}{d\P}\right]y.
\end{align*}
Using that due to~\eqref{eq:2} the supremum over $y>0$ in  the last expression is $e^x$ for $x=\E_{\QQ}\left[\lambda F-\log \frac{d\QQ}{d\P}\right]$  yields the result.
\end{proof}

The above duality estimate turns out to be tight in the sense that there is no duality gap  at least for discrete-time settings:
\begin{thm}\label{thm:duality}
 For $\TT=\{0=t_0<\dots<t_N\leq T\}$, let $\cA(\TT)$ be the set of $(\cG_t)_{t \in \TT}$-predictable strategies that only trade at times from $\TT$ and let $\cQ(\TT)$ denote the set of measures $\QQ\sim \P$ with bounded entropy relative to $\P$ that exhibit the relaxed martingale property~\eqref{eq:51} for $S$. Suppose in addition that there is $\epsilon>0$ such that $\E_{\P}[\exp(\lambda(1+\epsilon) |F|)]<\infty$ and $\E_{\P}[\exp(\epsilon|S_t|^{1+\epsilon})]<\infty$, $t \in \TT$. Finally, assume absence of arbitrage in the relaxed sense that $\cQ(\TT)\not=\emptyset$.

 Then the certainty equivalent of $F$ has the dual representation
  \begin{align}\label{eq:61}
    \inf_{\gamma \in \cA(\TT)} \frac{1}{\lambda} \log \E_{\P}\left[\exp\left(-\lambda(V^\gamma_T-F)\right)\right] = \sup_{\QQ \in \cQ(\TT)} \E_{\QQ}\left[F-\frac{1}{\lambda} \log \frac{d\QQ}{d\P}\right]
 \end{align}
 and the latter supremum is actually attained.
\end{thm}
\begin{rem}
The literature on exponential utility maximization mostly considers (locally) bounded asset prices and claims and provides examples for non-existence of an entropy-minimizing martingale measure if these conditions are violated; see, e.g., \cite{Fritelli:00}. In particular it does not cover the discrete-time version of our Bachelier model. Neither does this literature cover information delay. So, even though the above result is not particularly surprising, we include it for the sake of completeness.
\end{rem}
\begin{proof}
That `$\geq$' holds true in~\eqref{eq:61} is immediate from Lemma~\ref{lem:easydualitybound}.
Let us next observe that it is sufficient to treat the case $\lambda =1$ by passing to $\lambda F$ and that without loss of generality we can assume $F=0$ by passing from $\P$ to $\P^F$ with $d\P^F/d\P = e^{\lambda F}/\E_{\P}[e^{\lambda  F}]$ for which  $$\E_{\QQ}\left[\log \frac{d\QQ}{d\P^F}\right]=\E_{\QQ}\left[\log \frac{d\QQ}{d\P}-\lambda F\right],  \quad \QQ \in \cQ(\TT);$$
moreover, we have $\E_{\P^F}[\exp(\epsilon'|S_t|^{1+\epsilon'})] <\infty$ for any $\epsilon' \leq \epsilon/(1+\epsilon)$ by virtue of H{\"o}lder's inequality and the exponential integrability assumptions on $S$ and $F$ under $\P$.

To prove that the supremum in~\eqref{eq:61} is attained, we use the well-known Komlos-argument to obtain a maximizing sequence $\QQ_n \in \cQ(\TT)$, $n=1,2,...$, for which $Z_n := d\QQ_n/d\P$ converges almost surely; in fact, convergence holds also in $L^1(\P)$ because $H(\QQ_n|\P):=\E_\P\left[Z_n\log Z_n\right]=\E_{\QQ_n}\left[\log \frac{d\QQ_n}{d\PP}\right]$ can be assumed bounded in $n=1,2,\dots$ without loss of generality. Hence, $Z_0 := \lim_n Z_n$ yields the density of a probability measure $\QQ_0$ which by Fatou's lemma satisfies $H(\QQ_0|\P)\leq \lim_n H(\QQ_n|\P)$. So, attainment of the supremum in~\eqref{eq:61} is proven if we can argue that $\QQ_0$ is a martingale measure for $S$ in the relaxed sense of~\eqref{eq:51}. For this we show below that
\begin{align}\label{eq:4}
    R_M:= \max_{t \in \TT} \sup_{n=0,1,\dots}  \E_{\QQ_n}[|S_t|1_{\{|S_t|>M\}}] \to 0 \text{ as } M\uparrow \infty.
\end{align}
We then can estimate, for any $A \in \cF$ and $t \in \TT$,
\begin{align*}
    \left|\E_{\QQ_n}\left[S_t1_A\right]\right.&\left.-\E_{\QQ_0}\left[S_t1_A\right]\right| \\& \leq  \left|\E_{\QQ_n}\left[(-M \vee S_t \wedge M)1_A\right]-\E_{\QQ_0}\left[(-M \vee S_t \wedge M)1_A\right]\right|+ 2R_M
\end{align*}
to conclude from $Z_n \rightarrow Z_0$ in $L^1(\P)$ that
\begin{align*}
    \E_{\QQ_n}\left[S_t1_A\right] \xrightarrow[n \uparrow \infty]{} \E_{\QQ_0}\left[S_t 1_A\right], \quad A \in \cF.
\end{align*}
Exploiting the relaxed martingale property of $S$ under each of the $\QQ_n$, we conclude for any $t \in \TT$ and $A \in \cG_t$ the required martingale identity
\begin{align*}
    \E_{\QQ_0}[S_T1_A]=\lim_n \E_{\QQ_n}[S_T1_A]= \lim_n \E_{\QQ_n}[S_t1_A] =  \E_{\QQ_0}[S_t1_A].
\end{align*}

It remains to prove~\eqref{eq:4}. For this we let $q > 1$ be the conjugate of $p>1$ with $1/p+1/q=1$ to apply H{\"o}lder's inequality:
\begin{align}
    \E_{\QQ_n}[|S_t|1_{\{|S_t|>M\}}] &= \E_{\P}[(Z_n)^{1/p}|S_t| (Z_n)^{1/q}1_{\{|S_t|>M\}}] \nonumber \\&\label{eq:5}
    \leq (\E_{\P}[Z_n |S_T|^p])^{1/p}(\E_{\QQ_n}[1_{\{|S_t|>M\}}])^{1/q}.
\end{align}
Now, we recall $\epsilon>0$ from  our exponential integrability assumptions on $|S|^{1+\epsilon}$ and choose $p:=1+\epsilon$. Due to the Legendre-Fenchel duality estimate~\eqref{eq:2}, we have
\begin{align*}
    Z_n \epsilon |S_t|^p \leq e^{\epsilon |S_t|^{1+\epsilon}} + Z_n (\log Z_n-1)
\end{align*}
and we can thus bound the first expectation in~\eqref{eq:5} uniformly in $n=0,1,\dots$. Due to $\E_{\QQ_n}[1_{\{|S_t|>M\}}] \leq \E_{\QQ_n}[|S_t|]/M=\E_{\P}[Z_n |S_t|]/M$, the second one vanishes uniformly in $n=0,1,\dots$ for $M \uparrow \infty$ because we already have shown that $(Z_n|S_t|)_{n=1,2,\dots}$  is bounded in $L^1(\P)$, which by Fatou's lemma also yields $\E_{\QQ_0}[|S_t|]=\E_{\P}[Z_0|S_t|]<\infty$.  Since $\TT$ is finite, it follows that $\lim_M R_M = 0$ as claimed.

The proof of `$\leq$' in~\eqref{eq:61} can now be done following the approach in \cite{Fritelli:00} that also works with the relaxed martingale condition rather than the usual one. Indeed, as outlined in the proof of Theorem~2.3 there, the first order condition for the entropy minimizing measure reveals its density to be of the form $Z_0=\exp(-f_0)/\E_{\P}[\exp(-f_0)]$ for some random variable $f_0$ with $\E_{\QQ_0}[f_0]=0$ and $\E_\QQ[f_0]\leq 0$ for any $\QQ \in \cQ(\TT)$. The separation argument for Theorem~2.4 of \cite{Fritelli:00} then shows that $f_0$ is contained in the $L^1(\QQ)$-closure of $\{V^\gamma_T-L^\infty_+\;|\; \gamma \text{ bounded, $(\cG_t)_{t \in \TT}$ predictable}\}$ for any $\QQ \in \cQ(\TT)$. In particular there are $\gamma_n$ as considered in this set and $R_n \in L^\infty_+$, $n=1,2,\dots$, such that $f_0=\lim_n (V^{\gamma_n}_T-R_n)$ in probability. Since $\cQ(\TT)\not=\emptyset$, Theorem~1 in \cite{KS:06} yields that $f_0$ must be of the same form $f_0=V^{\gamma_0}_T-R_0$ for some $(\cG_t)_{t \in \TT}$-predictable $\gamma_0$ and some $R_0 \geq 0$. As a result, the left-hand side in~\eqref{eq:61} is less than or equal to
\begin{align*}
    \log \E_{\P}[\exp(- V^{\gamma_0}_T)] \leq  \log \E_{\P}[\exp(-f_0)].
\end{align*}
By construction of $\QQ_0$ and $f_0$ the right-hand side is equal to
\begin{align*}
    - \E_{\QQ_0} [\log Z_0] =\E_{\QQ_0}\left[f_0+\log \E_{\P}[\exp({-f_0})]\right]=\log \E_{\P}[\exp({-f_0})],
\end{align*}
which coincides with the upper bound we just have found for the left hand side. This yields the desired `$\leq$'-relation in~\eqref{eq:61}.
\end{proof}

\subsection{Lower bound for the scaling limit}\label{sec:lowerbound}

To prove `$\geq$' in~\eqref{2.3} we construct for vanishing time delay $H_n \downarrow 0$ relaxed martingale measures $\QQ_n$, $n=1,2,\dots$, under which $S_T$ converges in law to $s+\int_0^T \nu_u dW_u$ while their relative entropy with respect to $\P$ explodes at most like $\E_{\P}\left[\int_0^T(\nu_u-\sigma)^2du/(2\sigma^2H_n)\right]$.

To this end observe first that, by standard density arguments, we will get the same supremum there if instead of letting $\nu$ vary over all of $\Gamma$ there, we confine it to be   of the form
$\nu = \hat{\nu}(W)$ with
\begin{align*}
\hat{\nu}_t(x) := \sum_{j=1}^{J-1} f_j(x_{t_0},\dots,x_{t_{j-1}}) 1_{[t_{j},t_{j+1})}(t),  \quad t\in [0,T], \; x \in C[0,T],
  \end{align*}
where $0=t_0<t_1<\dots<t_J=T$ is a finite deterministic partition of $[0,T]$
and each $f_j:\mathbb R^{j}\rightarrow\mathbb R$, $j=1,...,J$, is Lipschitz continuous and bounded. Thus, it suffices to show that each such $\nu$ gives a lower bound for the scaling limit in~\eqref{2.3} along a fixed sequence $H_n \downarrow 0$:
\begin{equation}\label{3.6}
 \liminf_{n\rightarrow\infty} \pi(H_n,A/H_n,f(S_T))\geq
\mathbb E_{\mathbb P}\left[ f\left(s+\int_{0}^T \nu_t dW_t\right)- \frac{1}{2A\sigma^2}\int_{0}^T\left(\nu_t-\sigma\right)^2dt \right].
\end{equation}
For $n=1,2,\dots$, consider the SDE
\begin{align*}
    X^{n}_0=0,\quad
dX^{n}_t=dW_t+\kappa^n_t(X^n)dt
\end{align*}
with the path-dependent drift-coefficient
\begin{align*}
\kappa^n_t(x) &:= \frac{1}{\sigma}\left(\mu-\sum_{j=1}^{J-1}1_{[t_{j},t_{j+1})}(t)\frac{f_j(x_{t_0},\dots,x_{t_{j-1}})-\sigma}{H_n}\Phi\left(x_{t}-x_{(t-H_n)\vee t_{j}}\right)\right)
\end{align*}
where $\Phi(z):=(-1)\vee z \wedge 1$. Since $\kappa^n$ is Lipschitz continuous in the supremum norm on $C[0,T]$, Theorem~2.1 from Chapter IX in \cite{RY} yields a unique strong solution for this SDE. Moreover, boundedness of $\kappa^n$ ensures via Girsanov's theorem that
\begin{align}\label{eq:22}
    \frac{d\QQ_n}{d\P}:= \cE\left(-\int_0^. \kappa^n_u(X^n)dW_u\right)_T
\end{align}
yields the density of a measure $\QQ_n \sim \P$ under which $X^n$ is a Browninan motion. Hence,
\begin{align}\label{eq:21}
    S_t = s+\sigma X^n_t + \int_0^t  (\mu-\sigma \kappa^n_u(X^n)) du
\end{align}
is $\QQ_n$-integrable. Without loss of generality we can assume that $t_{j-1}\leq (t_j-H_n)^+$ for any $j=1,\dots,J$ and so, for any $t\in[t_{j},t_{j+1}]$, we get furthermore
\begin{align*}
    \E_{\QQ_n}&\left[S_{t_{j+1}}-S_t\middle| \cG^{H_n}_t \right] \\=& \int^{t_{j+1}}_t\frac{f_j(X^n_{t_0},\dots,X^n_{t_{j-1}})-\sigma}{H_n}\E_{\QQ_n}\left[ \Phi\left(X^n_{u}-X^n_{(u-H_n)\vee t_{j}}\right) \middle| \cF^W_{(t-H_n)^+}  \right]du.
\end{align*}
By symmetry of both $\Phi$ and the increments of Brownian motion, the latter conditional expectations all vanish and it thus follows that each $\mathbb Q_n$ exhibits the relaxed martingale property~\eqref{eq:51} for $\TT=[0,T]$ and $\cG_t=\cG^{H_n}_t$, $t \in [0,T]$.

Next, from \eqref{eq:21} and letting $\Psi(z):=\Phi(z)-z$, $z \in \RR$, we obtain that, for any $j=0,\dots,J-1$,
\begin{align*}
S_{t_{j+1}}-S_{t_j}=\;&
\sigma (X^{n}_{t_{j+1}}-X^{n}_{t_{j}})+\frac{\hat\nu_{t_j}(X^n)-\sigma}{H_n}\int_{t_j}^{t_{j+1}} \left(X^{n}_u-X^{n}_{(u-H_n)\vee t_j}\right) du\\
&+\frac{\hat\nu_{t_j}(X^n)-\sigma}{H_n}\int_{t_j}^{t_{j+1}} \Psi\left(X^{n}_u-X^{n}_{(u-H_n)\vee t_j}\right) du\\
=\;&\hat\nu_{t_j}(X^n) (X^{n}_{t_{j+1}}-X^{n}_{t_{j}})\\
&+\frac{\hat\nu_{t_j}(X^n)-\sigma}{H_n}\int_{t_{j+1}-H_n}^{t_{j+1}} (X^{n}_u-X^{n}_{t_{j+1}}) du\\
&+\frac{\hat\nu_{t_j}(X^n)-\sigma}{H_n}\int_{t_j}^{t_{j+1}} \Psi\left(X^{n}_u-X^{n}_{(u-H_n)^{+}}\right) du.
\end{align*}
Now take the sum over $j$. Since $X^n$ is a $\QQ_n$-Brownian motion and $|\Psi(z)| \leq z^4$, the contributions from the last two summands converge to 0 in distribution and the contribution from the first one converges in distribution to $\int_0^T \hat{\nu}_t(W)dW_t$. Hence, we can use Slutsky's theorem to conclude that
\begin{align*}
    \Law(S_T|\QQ_n) \to \Law\left(s+\int_0^T \nu_t dW_t\;\middle|\;\P\right)
\end{align*}
and, so, Fatou's lemma yields the lower bound
\begin{equation}\label{3.9}
\liminf_{n\rightarrow\infty}\mathbb E_{\mathbb Q_n}[f(S_T)]\geq\mathbb E_{\mathbb P}\left[f\left(s+\int_{0}^T\nu_t dW_t\right)\right].
\end{equation}

Finally, recalling the construction~\eqref{eq:22} of $\QQ_n$ we get
\begin{align*}
    \lim_n & \left(H_n\mathbb E_{\mathbb Q_n}\left[\log\left(\frac{d\mathbb Q_n}{d\mathbb P}\right)\right]\right)
=\lim_n \left(\frac{H_n}{2}\mathbb E_{\mathbb Q_n}\left[\int_{0}^T |\kappa^{(n)}_t|^2 dt\right]\right)\\
&=\lim_n \left(\sum_{j=1}^{J-1}
\mathbb E_{\mathbb Q_n}\left[\frac{(\hat\nu_{t_j}(X^n)-\sigma)^2 }{2\sigma^2 H_n}\int_{t_j}^{t_{j+1}}\mathbb E_{\mathbb Q_n}\left[\Phi^2\left(X^{n}_t-X^{n}_{(t-H_n)\vee t_j}\right)\middle|\mathcal F^W_{t_j}\right]dt\right]\right)\\
&=\lim_n \left(\sum_{j=1}^{J-1}
\mathbb E_{\mathbb P}\left[\frac{(\hat\nu_{t_j}(W)-\sigma)^2 }{2\sigma^2}\int_{t_j}^{t_{j+1}}\mathbb E_{\mathbb P}\left[(W_1)^2 \wedge \frac{1}{H_n}\right]dt\right]\right)\\
&=\sum_{j=1}^{J-1}\mathbb E_{\mathbb P}\left[\frac{(\nu_{t_j}-\sigma)^2}{2\sigma^2}(t_{j+1}-t_j)\right]
\end{align*}
where the last but one equality holds because $X^n$ is a $\QQ_n$-Brownian motion.
We conclude that
\begin{align}\label{eq:333}
\lim_{n\rightarrow\infty}\left(H_n\mathbb E_{\mathbb Q_n}\left[\log\left(\frac{d\mathbb Q_n}{d\mathbb P}\right)\right]\right)=
\frac{1}{2\sigma^2}\mathbb E_{\mathbb P}\left[\int_{0}^T\left(\nu_t-\sigma\right)^2dt \right].
\end{align}
This together with Lemma~\ref{lem:easydualitybound}
and \eqref{3.9} gives the desired lower bound \eqref{3.6} that we wanted to establish in this section.

\subsection{Upper bound for the scaling limit}\label{sec:upperbound}

For the upper bound of the scaling limit, we will need a lower bound for the relative entropy of a relaxed martingale measure with respect to $\P$ that we establish first:

\begin{lem}\label{lem:entropybound}
 For any delay $H>0$ and $M=1,2,\dots$, the relative entropy of a measure $\QQ \sim \P$ satisfying the relaxed martingale condition~\eqref{eq:51} over $\TT_{H/M}:=\{kH/M:k=0,1,\dots\} \cap [0,T]$ is bounded from below by
 \begin{align*}
     \E_{\QQ}\left[\log \frac{d\QQ}{d\PP}\right] \geq \frac{1}{2\sigma^2H} \frac{M}{M+1}  \E_{\QQ}\left[\left(S_{T_{H/M}}-s-\sigma W^\QQ_{T_{H/M}}\right)^2\right]
 \end{align*}
 where $T_{H/M}:=\max \TT_{H/M}$ and where $W^\QQ$ is the $\QQ$-Brownian motion induced by $W$ via Girsanov's theorem.
\end{lem}
\begin{proof}
 By It\^{o}'s representation theorem $W^\QQ = W-\int_0^. \alpha_u du$ for the predictable $\alpha$ such that
 \begin{align*}
     \left.\frac{d\QQ}{d\PP}\right|_{\cF^W_t}=\cE\left(\int_0^. \alpha_u dW_u\right)_t, \quad t \in [0,T].
 \end{align*}
 Hence, the relative entropy to be bounded from below is
 \begin{align*}
     \E_{\QQ}\left[\log \frac{d\QQ}{d\PP}\right] =\frac{1}{2}\E_{\QQ}\left[\int_0^T \alpha^2_u du\right]
 \end{align*}
 and
 \begin{align*}
     S_t-s-\sigma W^\QQ_t = \int_0^t \sigma \alpha_u du + \mu t, \quad t \in [0,T].
 \end{align*}

Fix $M\in\{1,2,\dots\}$ and $H>0$ and let us consider the increments over periods of length $H/M$:
\begin{align*}
    \Delta_t&:=S_{(t+H/M) \wedge
    T}-S_{t\wedge
    T}-\sigma\left(W^\QQ_{(t+H/M)\wedge
    T}-W^\QQ_{t\wedge
    T}\right)\\&= \int_{t\wedge
    T}^{(t+H/M)\wedge
    T}  (\sigma\alpha_u+\mu) du.
\end{align*}
Observe that
\begin{align}\label{eq:41}
S_{T_{H/M}}-s-\sigma W^{\QQ}_{T_{H/M}}=\sum_{k=1}^{K_{H/M}}    \Delta_{(k-1)H/M}
\end{align}
for $K_{H/M}:=T_{H/M}/(H/M)$. Moreover, due to the relaxed martingale property of $S$ (and the standard martingale property of $W^{\QQ}$) under $\QQ$, we have
\begin{align}\label{eq:42}
    \E_{\QQ}[\Delta_t]=0 \quad \text{ and } \quad  \E_{\QQ}[\Delta_t\Delta_s]=0
\end{align}
for any $t \in [0,T]$ and any $s \in [0,T]$ with $|t-s|\geq (M+1)H/M$.

Hence, subdividing $[0,T]$ into intervals of length $H/M$ we can estimate by Jensen's inequality
\begin{align*}
      \E_{\QQ}\left[\log \frac{d\QQ}{d\PP}\right] &= \frac{1}{2}\sum_{k=1}^\infty  \E_{\QQ}\left[\int_{((k-1)H/M)\wedge T}^{(kH/M)\wedge T} \alpha^2_u du\right] \\&\geq \frac{1}{2}\sum_{k=1}^{K_{H/M}} \E_{\QQ}\left[\frac{1}{H/M}\left(\int_{((k-1)H/M}^{kH/M} \alpha_u du\right)^2\right]\\
      &= \frac{M}{2H}\sum_{k=1}^{K_{H/M}} \E_{\QQ}\left[\left(\frac{\Delta_{(k-1)H/M}-\mu H/M}{\sigma}\right)^2\right]
      \\&=  \frac{M}{2H\sigma^2}\sum_{k=1}^{K_{H/M}}\left( \E_{\QQ}\left[\left(\Delta_{(k-1)H/M}\right)^2\right]+(\mu H/M)^2\right),
  \end{align*}
where in the last step we used $\E_{\QQ}[\Delta^M_t]=0$. Dropping the squares involving $\mu$ and rearranging the sum over $k$ gives
\begin{align*}
   \E_{\QQ}\left[\log \frac{d\QQ}{d\PP}\right] & \geq  \frac{M}{2H\sigma^2}\sum_{i=0}^M \sum_{j:i+j(M+1)< K_{H/M}} \E_{\QQ}\left[\left(\Delta_{(i+j(M+1)H/M}\right)^2\right].
\end{align*}
In this rearrangement, we have due to~\eqref{eq:42} that for each $i=0,\dots,M-1$ the sums over $j$  are taken over $\QQ$-uncorrelated increments and therefore amount to
\begin{align*}
    \E_{\QQ}\left[\sum_{j} \left(\Delta_{(i+j(M+1))H/M)}\right)^2\right]
    =\E_{\QQ}\left[\left(\sum_{j} \Delta_{(i+j(M+1))H/M)}\right)^2\right].
\end{align*}
Plugging this into our above estimate yields
\begin{align*}
    \E_{\QQ}\left[\log \frac{d\QQ}{d\PP}\right] & \geq  \frac{M}{2H\sigma^2}  \E_{\QQ}\left[\sum_{i=0}^M\left(\sum_{j:i+j(M+1)< K_{H/M}}\Delta_{(i+j(M+1)H/M}\right)^2\right]\\
    & \geq \frac{M}{2H\sigma^2}   \E_{\QQ}\left[\frac{1}{M+1}\left(\sum_{i=0}^M\sum_{j:i+j(M+1)< K_{H/M}}\Delta_{(i+j(M+1)H/M}\right)^2\right]
\end{align*}
where the last estimate is due to Jensen's inequality. The above double sum is aggregating all increments $\Delta$ over intervals $[(k-1)H/M,kH/M]$ fully fitting into $[0,T]$ and thus coincides with the left-hand side of~\eqref{eq:41}. We obtain the claimed lower bound for the entropy.
\end{proof}

Let us now \textbf{bound the scaling limit in~\eqref{2.3} from above}. For this, observe that, due to Remark~\ref{rem:denom}, the indifference price for $f(S_T)$ behaves for exploding risk aversion $\lambda=A/H$ like the certainty equivalent
\begin{align*}
    c_H:=\frac{1}{A/H} \log \inf_{\gamma} \E_{\P}\left[\exp\left(-\frac{A}{H}(V^\gamma_T-f(S_T))\right)\right].
\end{align*}
By standard density arguments, $c_H$ is continuous in $H>0$. When taking the scaling limit $H \downarrow 0$, we can thus restrict attention to $H>0$ such that $T/H \in \{1,2,\dots\}$.
If we confine ourselves to $(\cG^H_t)$-predictable strategies $\gamma$ that only intervene at times in $\TT_{H/M}$ as defined in the previous lemma, the above certainty equivalent can only grow and we can use the discrete-time duality Theorem~\ref{thm:duality} to estimate it by
\begin{align*}
    c_H \leq \E_{\QQ_{H/M}}\left[f(S_T)-\frac{1}{A/H} \log \frac{d\QQ_{H/M}}{d\P}\right]
\end{align*}
where $\QQ_{H/M} \sim \P$ is a probability that exhibits the relaxed martingale property~\eqref{eq:51} for $\TT=\TT_{H/M}$ and that has finite entropy relative to $\P$. This entropy can be bounded from below by  Lemma~\ref{lem:entropybound} and we thus obtain
\begin{align}\label{referee}
    c_H \leq \E_{\QQ_{H/M}}\left[f(S_T)-\frac{1}{A/H} \frac{M}{2H\sigma^2(M+1)}\left(S_{T}-s-\sigma W^{\QQ_{H/M}}_{T}\right)^2\right]
\end{align}
for the $\QQ_{H/M}$-Brownian motion $W^{\QQ_{H/M}}$ induced by $W$ via Girsanov's theorem; here we used that $T=T_{H/M} = \max \TT_{H/M}$ which holds because $T/H \in \{1,2,\dots\}$ by our assumption on $H$.

Now, if $S_T$ is measurable with respect to $\sigma(W^{\QQ_{H/M}}_u, \; u \in [0,T])$, we can apply the martingale representation theorem to find a predictable $\nu \in L^2(\P \otimes du)$ with $f(S_T) = s + \int_0^T \nu_u dW^{\QQ_{H/M}}_u$; an upper bound in the form required by~\eqref{2.3} emerges then from~\eqref{referee} by use of It\^{o}'s isometry and letting $M \uparrow \infty$.
In general, though, $S_T$ is just $\cF^W_T$-measurable. Following Theorem~1 in~\cite{S:76}, we exploit that our probability space $(\Omega,\cF,\QQ_{H/M})$ is rich enough to support a continuously distributed random variable independent of $W^{\QQ_{H/M}}_T$  such as  $U:=W^{\QQ_{H/M}}_T-2W^{\QQ_{H/M}}_{T/2}$ to find a Borel-measurable function $\mathbf{s}:\RR^2 \to \RR$ for which
\begin{align*}
    \Law\left((S_T,W^{\QQ_{H/M}}_T)\;\middle|\;\QQ_{H/M}\right) = \Law\left((\mathbf{s}(U,W^{\QQ_{H/M}}_T),W^{\QQ_{H/M}}_T)\;\middle|\;\QQ_{H/M}\right).
\end{align*}
This allows us to continue our estimate (\ref{referee}) as follows:
\begin{align*}
    c_H &\leq \E_{\QQ_{H/M}}\left[f(\mathbf{s}(U,W^{\QQ_{H/M}}_T))- \frac{M}{2A\sigma^2(M+1)}\left(\mathbf{s}(U,W^{\QQ_{H/M}}_T)-s-\sigma W^{\QQ_{H/M}}_{T}\right)^2\right]
    \\&=\E_{\P}\left[f(\mathbf{s}(U',W_T))- \frac{M}{2A\sigma^2(M+1)}\left(\mathbf{s}(U',W_T)-s-\sigma W_{T}\right)^2\right]
\end{align*}
where $U':=W_T-2W_{T/2}$ is independent of $W_T$ under $\P$ with the same distribution as $U$ under $\QQ_{H/M}$. Clearly, the fact that $S_T$ is square-integrable under $\QQ_{H/M}$ with expectation $s$, translates into $\mathbf{s}(U',W_T)$ having the same properties under $\P$. So, by It\^{o}'s representation theorem it is of the form $s+\int_0^T \nu_u dW_u$ for some $(\cF^W_t)$-predictable $\nu \in L^2(\P\otimes dt)$ and we obtain that
\begin{align*}
    c_H &\leq \E_{\P}\left[f\left(s+\int_0^T \nu_t dW_t\right)- \frac{M}{2A\sigma^2(M+1)}\left(\int_0^T (\nu_t-\sigma) dW_t\right)^2\right].
\end{align*}
Letting $M\uparrow \infty$ and using It\^{o}'s isometry yields the desired upper bound.

\bibliography{finance}{}
\bibliographystyle{abbrv}

\end{document}